%% file: activityarray-full.tex
\documentclass[10pt,letterpaper]{article}

\input{settings-custom.tex}

\begin{document}
\input{titlepage.tex}

\input{intro.tex}
\input{model-problems.tex}

\input{related-work.tex}

\input{algorithm.tex}
\input{analysis-poly.tex}

\input{weak-infinity-analysis.tex}

\input{results.tex}
\input{conclusion.tex}

\input{biblio.tex}

\end{document}

%% file: settings-custom.tex
\usepackage[boxed,nofillcomment,linesnumbered,noend,noresetcount, noline]{algorithm2e}
\usepackage{times}
\usepackage{fullpage}
\usepackage{cite}
\usepackage{graphicx}
\usepackage{epsf}
\usepackage{tabularx}
\usepackage{multirow}
\usepackage[small, center]{caption}
\usepackage{amsmath}
\usepackage{amsthm}
\usepackage{amssymb,latexsym,color}
\usepackage{mdwlist}

\usepackage{caption}
\usepackage{subcaption}

\newcommand{\idlow}[1]{\mathord{\mathcode`\-="702D\it #1\mathcode`\-="2200}}
\newcommand{\id}[1]{\ensuremath{\idlow{#1}}}
\newcommand{\litlow}[1]{\mathord{\mathcode`\-="702D\sf #1\mathcode`\-="2200}}
\newcommand{\lit}[1]{\ensuremath{\litlow{#1}}}

\newtheorem{lemma}{Lemma}
\newtheorem{theorem}{Theorem}
\newtheorem{claim}{Claim}
\newtheorem{corollary}{Corollary}
\newtheorem{prop}{Proposition}

\newcommand{\Get}{\lit{Get}}

\newcommand{\Free}{\lit{Free}}
\newcommand{\Call}{\lit{Call}}
\newcommand{\Collect}{\lit{Collect}}
\newcommand{\LA}{\lit{LevelArray}}

\usepackage[textsize=tiny]{todonotes}

\newenvironment{proofT}[1]{\proof\def\toto{#1}}{\hspace*{\fill}$\Box_{Theorem~\ref{\toto}}$\par\vspace{3mm}}

\newenvironment{proofL}[1]{\proof\def\toto{#1}}{\hspace*{\fill}$\Box_{Lemma~\ref{\toto}}$\par\vspace{3mm}}
\newenvironment{proofC}{\proof}{\hspace*{\fill}$\Box_{Corollary~\ref{\toto}}$\par\vspace{3mm}}

\newenvironment{proofP}[1]{\proof\def\toto{#1}}{\hspace*{\fill}$\Box_{Proposition~\ref{\toto}}$\par\vspace{3mm}}

\newcommand{\thmpostponed}[2]
{
\newcounter{#1}
\setcounter{#1}{\value{theorem}}
\begin{theorem}
\label{thm:#1}
#2
\end{theorem}
\expandafter\def\csname #1\endcsname{
\newcounter{#1temp}
\setcounter{#1temp}{\value{theorem}}
\setcounter{theorem}{\value{#1}}
\begin{theorem}
#2
\end{theorem}
\setcounter{theorem}{\value{#1temp}}
}
\vspace{-1.5em}
}

\newcommand{\thmpostponedwname}[3]
{
\newcounter{#1}
\setcounter{#1}{\value{theorem}}
\begin{theorem}[#2]
\label{thm:#1}
#3
\end{theorem}
\expandafter\def\csname #1\endcsname{
\newcounter{#1temp}
\setcounter{#1temp}{\value{theorem}}
\setcounter{theorem}{\value{#1}}
\begin{theorem}[#2]
#3
\end{theorem}
\setcounter{theorem}{\value{#1temp}}
}
\vspace{-1.5em}
}

\newcommand{\lemmaproofpostponedwname}[4]
{
\newcounter{#1}
\newcounter{#1temp}
\setcounter{#1}{\value{lemma}}
\begin{lemma}[#2]
\label{#1}
#3
\end{lemma}
\expandafter\def\csname #1\endcsname{
\setcounter{#1temp}{\value{lemma}}
\setcounter{lemma}{\value{#1}}
\begin{lemma}
#3
\end{lemma}
\setcounter{theorem}{\value{#1temp}}
\begin{proofL}{#1}
#4
\end{proofL}
}
}

\newcommand{\thmproofpostponedwname}[4]
{
\newcounter{#1}
\newcounter{#1temp}
\setcounter{#1}{\value{theorem}}
\begin{theorem}[#2]
\label{#1}
#3
\end{theorem}
\expandafter\def\csname #1\endcsname{
\setcounter{#1temp}{\value{theorem}}
\setcounter{theorem}{\value{#1}}
\begin{theorem}
#3
\end{theorem}
\setcounter{theorem}{\value{#1temp}}
\begin{proofT}{#1}
#4
\end{proofT}
}
}

\newcommand{\lemmapostponed}[2]
{
\newcounter{#1}
\newcounter{#1temp}
\setcounter{#1}{\value{theorem}}
\begin{lemma}
\label{#1}
#2
\end{lemma}
\expandafter\def\csname #1\endcsname{
\setcounter{#1temp}{\value{theorem}}
\setcounter{theorem}{\value{#1}}
\begin{lemma}
#2
\end{lemma}
\setcounter{theorem}{\value{#1temp}}
}
\vspace{-1.5em}
}

\newcommand{\lemmapostponedwname}[3]
{
\newcounter{#1}
\newcounter{#1temp}
\setcounter{#1}{\value{theorem}}
\begin{lemma}[#2]
\label{#1}
#3
\end{lemma}
\expandafter\def\csname #1\endcsname{
\setcounter{#1temp}{\value{theorem}}
\setcounter{theorem}{\value{#1}}
\begin{lemma}[#2]
#3
\end{lemma}
\setcounter{theorem}{\value{#1temp}}
}
\vspace{-1.5em}
}

\newcommand{\lemmaproofpostponed}[3]
{
\newcounter{#1}
\newcounter{#1temp}
\setcounter{#1}{\value{lemma}}
\begin{lemma}
\label{#1}
#2
\end{lemma}
\expandafter\def\csname #1\endcsname{
\setcounter{#1temp}{\value{lemma}}
\setcounter{lemma}{\value{#1}}
\begin{lemma}
#2
\end{lemma}
\setcounter{theorem}{\value{#1temp}}
\begin{proofL}{#1}
#3
\end{proofL}
}
}

\newcommand{\corollaryproofpostponed}[3]
{
\newcounter{#1}
\newcounter{#1temp}
\setcounter{#1}{\value{corollary}}
\begin{corollary}
\label{#1}
#2
\end{corollary}
\expandafter\def\csname #1\endcsname{
\setcounter{#1temp}{\value{corollary}}
\setcounter{corollary}{\value{#1}}
\begin{corollary}
#2
\end{corollary}
\setcounter{corollary}{\value{#1temp}}
\begin{proofC}
#3
\renewcommand{\toto}{#1}
\end{proofC}
}
}

\newcommand{\propproofpostponed}[3]
{
\newcounter{#1}
\newcounter{#1temp}
\setcounter{#1}{\value{prop}}
\begin{prop}
\label{#1}
#2
\end{prop}
\expandafter\def\csname #1\endcsname{
\setcounter{#1temp}{\value{prop}}
\setcounter{prop}{\value{#1}}
\begin{prop}
#2
\end{prop}
\setcounter{prop}{\value{#1temp}}
\begin{proofP}{#1}
#3
\renewcommand{\toto}{#1}
\end{proofP}
}
}

\newcommand{\claimproofpostponed}[3]
{
\newcounter{#1}
\newcounter{#1temp}
\setcounter{#1}{\value{claim}}
\begin{claim}
\label{#1}
#2
\end{prop}
\expandafter\def\csname #1\endcsname{
\setcounter{#1temp}{\value{claim}}
\setcounter{claim}{\value{#1}}
\begin{claim}
#2
\end{claim}
\setcounter{prop}{\value{#1temp}}
\begin{claimCl}{#1}
#3
\renewcommand{\toto}{#1}
\end{claimCl}
}
}

\newcounter{casenum}

\renewcommand{\paragraph}[1]{\vspace{0.02cm} \noindent\textbf{#1}\hspace{0.15em}}

\newcommand{\toto}{xxx}

\SetFuncSty{textsf}
\SetFuncSty{small}

\newtheorem{definition}{Definition} 

%% file: titlepage.tex

\title{The LevelArray: A Fast, Practical  Long-Lived Renaming Algorithm}

\date{}

\author{ Dan Alistarh\footnote{Part of this work was performed while the author was a Postdoctoral Fellow at MIT CSAIL.} \\ \small{Microsoft Research Cambridge}
\and Justin Kopinsky \\ \small{MIT}
\and Alexander Matveev \\ \small{MIT}
\and Nir Shavit \\ \small{MIT and Tel-Aviv University}}

\maketitle

\begin{abstract}
The long-lived renaming problem appears in shared-memory systems where a set of threads need to register and deregister frequently from 
the computation, while concurrent operations scan the set of currently registered threads. Instances of this problem show up in concurrent implementations 
of transactional memory, flat combining, thread barriers, and memory reclamation schemes for lock-free data structures. 

In this paper, we analyze a randomized solution for long-lived renaming. 
The algorithmic technique we consider, called the \emph{LevelArray}, has previously been used for 
hashing and one-shot (single-use) renaming. 
Our main contribution is to prove that, in \emph{long-lived}  executions, where processes may register and deregister polynomially many times, the technique guarantees \emph{constant} steps on average and $O( \log \log n )$ steps with high probability for registering, \emph{unit} cost for deregistering, and $O(n)$ steps for collect queries, where $n$ is an upper bound on the number of processes that may be active at any point in time. 
We also show that the algorithm has the surprising property that it is \emph{self-healing}: under reasonable assumptions on the schedule, operations running while the data
structure is in a degraded state implicitly help the data structure re-balance itself. This subtle mechanism obviates the need for expensive periodic rebuilding procedures. 

Our benchmarks validate this approach, showing that, for typical use parameters, 
the average number of steps a process takes to register is less than \emph{two} and the worst-case number of steps is
bounded by \emph{six}, even in executions with billions of operations. We contrast this with other randomized
implementations, whose worst-case behavior we show to be unreliable, and with deterministic implementations, whose cost
is linear in $n$. 

%
\end{abstract}
%

%% file: intro.tex
\setcounter{page}{1}
\section{Introduction}
\label{sec:intro}

Several shared-memory coordination problems can be reduced to the following task: a set of threads dynamically register
and deregister from the computation, while other threads periodically query the set of registered threads. A standard
example is memory management for lock-free data structures, e.g.~\cite{DHLM11}: threads accessing the
data structure need to 
register their operations, to ensure that a memory location which they are accessing does not get freed while
still being addressed. Worker threads must register and deregister efficiently, while the ``garbage
collector" thread queries the set of registered processes periodically to see which memory locations can be freed. 
Similar mechanisms are employed in software transactional memory
(STM),  e.g.~\cite{DMS10,AfekMS12},  to detect conflicts between reader and writer threads, in flat
combining~\cite{HISM10} to determine which
threads have work to be performed, and in shared-memory barrier algorithms~\cite{HSBook}.  
In most applications, the time to complete  registration directly affects the performance of the method
calls that use it. 

Variants of the problem are known under different names: in a theoretical setting, it has been
formalized as \emph{long-lived renaming}, e.g.~\cite{MoirA95, BEW11, AttiyaF01}; in a practical setting, it
is known as
\emph{dynamic collect}~\cite{DHLM11}. Regardless of the name, requirements are similar: for good performance,
processes should register and deregister quickly, since these operations are very frequent.
Furthermore, the data structure should be space-efficient, and its performance should depend on the contention
level.

Many known solutions for this problem, e.g.~\cite{DHLM11, MoirA95}, are based on an approach we call the \emph{activity
array}.
Processes share a set of memory locations, whose size is in the order of the number of threads $n$. A thread
\emph{registers} by
acquiring a location through a \lit{test-and-set} or \lit{compare-and-swap} operation, and
\emph{deregisters} by
re-setting the location to its initial state. A \emph{collect} query simply scans the array to determine which processes
are currently
registered.\footnote{The trivial solution where a thread simply uses its identifier as the index of a unique array location is inefficient, 
since the complexity of the collect would depend on the size of the id space, instead of the maximal contention $n$.} 
The activity array has the advantage of relative simplicity and good performance in practice~\cite{DHLM11}, due
to the array's good cache behavior during collects.

One key difference between its  various 
implementations is the way the register operation is implemented. A simple strategy is to
scan the array from left to right, until the first free location is found~\cite{DHLM11}, incurring \emph{linear}
complexity on average. A more complex procedure is to probe locations chosen at random or based on a hash
function~\cite{AHM11, AAGGG10}, or to proceed by probing linearly from a randomly chosen location. The
expected step complexity of this second approach should be constant on average, and at least logarithmic in the worst
case. One disadvantage of known randomized approaches, which we also illustrate in our experiments, is that their worst-case performance is not stable over long executions: 
while most operations will be fast, there always exist operations which take a long time. 
Also, in the case of linear probing, the performance of the data structure is known to degrade over time, a phenomenon known as 
\emph{primary clustering}~\cite{McAllister08}. 

It is therefore natural to ask if there exist 
solutions which combine the \emph{good average performance} of randomized techniques with the 
of \emph{stable worst-case bounds} of deterministic algorithms.

In this paper, we show that such efficient solutions exist, by proposing a long-lived activity array with sub-logaritmic
worst-case time for registering, and stable worst-case behavior in practice. The algorithm, called \emph{LevelArray}, guarantees
\emph{constant} average complexity and $O( \log \log n)$ step complexity with high probability for registering,
\emph{unit} step complexity for deregistering, and \emph{linear} step complexity for collect queries. 
Crucially, our analysis shows that these properties
are guaranteed over \emph{long-lived} executions, where processes may register, deregister, and collect polynomially
many times against an oblivious adversarial scheduler. 

The above properties should be sufficient for good performance in long-lived executions. Indeed, even if the
performance of the data structure were to degrade over time, as is the case with  hashing
techniques, e.g.~\cite{McAllister08},
we could rebuild the data structure periodically, preserving the bounds in an amortized sense. However, our analysis 
shows that this explicit rebuilding mechanism is not necessary since the data structure is ``self-healing.''
 Under reasonable assumptions on the schedule, even if the data structure ends up in extremely unbalanced state 
 (possible during an infinite execution), the deregister and register operations
running from this state automatically re-balance the data structure with high probability. 
The self-healing property removes the need for explicit 
rebuilding. 

The basic idea behind the algorithm is simple, and has been used previously for efficient
hashing~\cite{BK90} and
\emph{one-shot}\footnote{One-shot renaming~\cite{ABDPR90} is the variant of the problem where processes only
register once, and deregistration is not possible.} randomized renaming~\cite{AAGW13}. We consider an array of size $2n$,
where $n$ is an upper bound on contention. We
split the locations into $O( \log n)$ \emph{levels}: the first (indexed by $0$) contains the first $3n / 2$ locations, the second
contains
the next $n / 4$ and so on, with the $i$th level containing $n / 2^{i}$ locations, for $i \geq 1$. To
register, each process
performs a  
\emph{constant} number of \lit{test-and-set} probes at each level, stopping the first time when it acquires a location.
Deregistering is performed by simply resetting the location, while collecting is done by scanning the $2n$ locations. 

This algorithm clearly solves the problem, the only question is its complexity. The intuitive reason
why this procedure runs in $O( \log \log n )$ time in a one-shot execution is that, as processes
proceed towards higher
levels, the number of
processes competing in a level $i$ is $O( n / 2^{2^i})$, while the space available is $\Theta (n / 2^i)$. By level
$\Theta(\log \log n)$, there are virtually no more processes competing. This intuition was formally captured and proven in \cite{BK90}. However, it was not clear if anything close to this efficient behavior holds true in the long-lived case where threads continuously register and deregister. 

The main technical contribution of our paper is showing that this procedure does indeed work in long-lived polynomial-length 
executions,
and, perhaps more surprisingly, requires no re-building over infinite executions, given an oblivious adversarial scheduler. 
The main challenge is in bounding the correlations between the processes' operations, and 
in analyzing the properties of the resulting probability distribution over the data structure's state. 
More precisely, we identify a ``balanced'' family of probability distributions over the level occupancy under which most operations are fast. 
 We then analyze sequences of operations of increasing length, 
and prove that they are likely to keep the data structure balanced, despite the fact that the scheduling and the process input may be correlated in 
arbitrary ways (see Proposition~\ref{prop:overcrowded}). One further difficulty comes from the fact that we allow 
the adversary to insert arbitrary sequences of operations between a thread's 
register and the corresponding deregister (see Lemma~\ref{timeinvariant}), as is the case in a real execution. 

The previous argument does not preclude the data
structure from entering an unbalanced state over an infinite execution. (Since
it has non-zero probability, such an event will eventually occur.) This motivates us to analyze such executions as
well. We show that, assuming the system schedules polynomially many steps between the time a process starts a
register operation and the time it deregisters,\footnote{This assumption prevents 
unrealistic schedules in which the adversary brings the data structure in an unbalanced state, and then schedules a
small set of threads to register and unregister infinitely many times, keeping the data structure in roughly the same
state while inducing high expected cost on the threads.}
the data structure will rebalance itself from an arbitrary initial state, with high probability. 

Specifically, in a bad state, the array may be
arbitrarily shifted away from this good distribution. We prove that, as more and more operations release slots and
occupy new ones, the data structure gradually shifts back to a good distribution, which is reached with high probability
after polynomially many system steps are taken. Since this shift must occur from \emph{any} unbalanced state, it follows
that, in fact, every state is well balanced with high probability. Finally, this implies that every operation verifies
the $O(\log \log n)$ complexity upper bound with high probability. 

From a theoretical perspective, the \lit{LevelArray} algorithm solves non-adaptive
long-lived renaming in $O( \log \log n )$ steps with
high probability, against an oblivious adversary in polynomial-length executions. The same guarantees are provided in
infinite executions under scheduler assumptions. 
We note that our analysis can also be extended to provide worst-case bounds on the long-lived performance of the Broder-Karlin hashing algorithm~\cite{BK90}. (Their analysis is one-shot, which is standard for hashing.)  

The algorithm is \emph{wait-free}. The logarithmic lower bound
of Alistarh et al.~\cite{AAGG11}
on the complexity of one-shot randomized adaptive 
renaming is circumvented since the algorithm is not namespace-adaptive. The algorithm is time-optimal for one-shot renaming when 
linear space and \lit{test-and-set} operations are used~\cite{AAGW13}. 

We validate this approach through several benchmarks. Broadly, the tests show that, for common use parameters, the data
structure guarantees fast registration---less than two probes on average for an array of size $2n$---and
that the performance is surprisingly stable when dealing with contention and long executions. To illustrate, in a
benchmark with approximately one billion register and unregister operations with $80$ concurrent threads, the maximum
number of probes performed by \emph{any} operation was \emph{six}, while the average number of probes for registering 
was around $1.75$.

The data structure compares favorably to other randomized and deterministic techniques.
In particular,
the worst-case number of steps performed is at least
an order of magnitude lower than that of any other implementation.
We also tested the ``healing'' property by initializing the data structure in a bad state and running a typical schedule
from that state. The data structure does indeed converge to a balanced distribution (see
Figure~\ref{fig:histogram}); interestingly, the convergence speed towards the good state is higher than
predicted by the analysis. 

\paragraph{Roadmap.} Section~\ref{sec:model} presents the system model and problem statement. Section~\ref{sec:rw}
gives an overview of related work. We present the algorithm in Section~\ref{sec:algorithm}. Section~\ref{sec:poly} gives
the analysis of polynomial-length executions, while Section~\ref{sec:infty-analysis} considers infinite executions. We
present the implementation results in Section~\ref{sec:results}, and conclude in Section~\ref{sec:conclusion}. 

%% file: model-problems.tex

\section{System Model and Problem Statement}
\label{sec:model}

We assume the standard asynchronous shared memory model with $N$ processes (or threads) $p_1, \ldots, p_N$, 
out of which at most $n \leq N$ participate in any execution. (Therefore, $n$ can be seen as an upper bound on the contention in an execution.)
To simplify the exposition, 
in the analysis, we will denote the $n$ participants by $p_1, p_2, \ldots, p_n$, although the identifier $i$ is unknown to process $p_i$ in the actual execution. 

Processes communicate by performing operations on shared registers, specifically \lit{read}, \lit{write}, \lit{test-and-set} or
\lit{compare-and-swap}. Our algorithm only employs \lit{test-and-set} operations. (\lit{Test-and-set}
operations can be
simulated either using reads and writes with randomization~\cite{AG92}, or atomic \lit{compare-and-swap}. 
Alternatively, we can use the adaptive \lit{test-and-set} construction of Giakkoupis and Woelfel~\cite{GW12} 
to implement our algorithm using only reads and writes with an extra multiplicative $O(\log^* n)$ factor in the running time.) We say that
a process \emph{wins} a \lit{test-and-set} operation if it manages to change the value of the location from $0$ to $1$.
Otherwise, it \emph{loses} the operation. The winner may later \emph{reset} the location by setting it back to $0$. 
 We assume that each process has a local random number generator,
accessible through the call $\lit{random}(1, v)$, which returns a uniformly random integer between $1$ and $v$. 

The processes' input and their scheduling are controlled by an \emph{oblivious adversary}. The adversary knows the
algorithm and the distributions from which processes draw coins, but does not see the results of the local coin flips or
of other operations performed during the execution. Equivalently, the adversary must decide on the complete schedule and input
\emph{before} the algorithm's execution.  

An \emph{activity array} data structure exports
three operations. The $\lit{Get}()$ operation returns a unique index to the process; $\lit{Free}()$
releases the index returned by the most recent $\lit{Get}()$, while $\lit{Collect}()$ returns a set of indices, 
such that any index held by a process throughout the $\lit{Collect}()$ call must be returned. 
The adversary may also require processes to take steps running arbitrary algorithms between activity array operations. We model this by allowing the adversary to introduce a $\Call{}()$ operation, which completes in exactly one step and does not read or write to the activity array. The adversary can simulate longer algorithms by inputting consecutive $\Call{}()$ operations.

The input for each process is \emph{well-formed}, in that \Get{} and \Free{} operations alternate, starting with a
\Get{}. 
\lit{Collect} and \Call{} operations may be interspersed arbitrarily. Both \lit{Get} and \lit{Free}
are required to be linearizable. We say a process \emph{holds} an index $i$ between the linearization points of the \lit{Get}
operation that returned $i$ and that of the corresponding $\lit{Free}$ operation. The key correctness property of the
implementation is that no two processes hold the same index at the same point in time. 
 \lit{Collect} must return a set of indices with the following \emph{validity} property: any index returned by
\lit{Collect} must have been held by some process
 during the execution of the operation. (This operation is not an atomic snapshot of the array.) 

From a theoretical perspective, an activity array implements the \emph{long-lived renaming}
problem~\cite{MoirA95}, where the \lit{Get} and \lit{Free} operations correspond to \lit{GetName} and
\lit{ReleaseName}, respectively. However, the \lit{Collect} operation additionally imposes the requirement that the
names should be enumerable efficiently. The namespace upper bound usually required for renaming can be
translated as an upper bound on the step
complexity of \lit{Collect} and on the space complexity of the overall implementation. 

We focus on the \emph{step complexity} metric, i.e. the number of steps that a process performs while executing an
operation. We say that an event occurs \emph{with high probability} (w.h.p.) if its probability is at least $1 - 1 /
n^\gamma$, for $\gamma \geq 1$ constant.

%% file: related-work.tex
\section{Related Work}
\label{sec:rw}


The \emph{long-lived renaming} problem was introduced by Moir and Anderson~\cite{MoirA95}. (A similar
variant of renaming~\cite{ABDPR90} had been previously considered by Burns and Peterson~\cite{BP89}.) Moir and
Anderson 
presented several deterministic algorithms, assuming various shared-memory primitives. In
particular, they introduced an array-based algorithm where each process probes $n$ locations linearly using
\lit{test-and-set}. A similar algorithm was considered in the context of $k$-exclusion~\cite{AM97}. 
A considerable amount of subsequent research, e.g.~\cite{Moir98, AST02, AF03, BEW11} studied faster deterministic
solutions for long-lived renaming. To the best of our knowledge, all these algorithms either
have linear or super-linear
step complexity~\cite{AST02, AF03, BEW11}, or employ strong primitives such as \lit{set-first-zero}~\cite{MoirA95},
which are not available in general. Linear time complexity is known to be inherent for deterministic renaming algorithms
which employ \lit{read}, \lit{write}, \lit{test-and-set} and \lit{compare-and-swap} operations~\cite{AAGG11}.
For a complete overview of known approaches for deterministic long-lived renaming 
we direct the reader to reference~\cite{BEW11}. Despite progress on the use of randomization for fast \emph{one-shot}
renaming, e.g.~\cite{AACGZ11}, no randomized algorithms for long-lived renaming were known prior to our work. 

The idea of splitting the space into levels to minimize the number of collisions was first used by Broder and
Karlin~\cite{BK90} in the context of hashing. Recently,~\cite{AAGW13} used a 
similar idea to obtain a \emph{one-shot} randomized loose renaming algorithm against a strong adversary. 
Both references~\cite{BK90, AAGW13} obtain expected worst-case complexity $O( \log \log n )$ in \emph{one-shot}
executions, consisting of exactly one \Get{} per thread and do not consider long-lived executions. 
In particular, our work can be seen as an extension of~\cite{AAGW13} for the long-lived case, against an oblivious adversary. 
Our analysis will imply the upper bounds of~\cite{BK90, AAGW13} in one-shot executions, 
as it is not significantly affected by the strong adversary in the one-shot case. 
However, our focus in this paper is analyzing \emph{long-lived} polynomial and infinite executions, and 
showing that the technique is viable in practice.

A more applied line of research~\cite{ROP, DHLM11} employed data structures similar to activity arrays in the context of
memory reclamation for lock-free data structures. One important difference from our approach is that the solutions
considered are deterministic, and have $\Omega(n)$ complexity for registering since processes perform probes
linearly.  The
\emph{dynamic collect} problem defined in~\cite{DHLM11} has similar semantics to those of the activity array,
and adds operations that are specific to memory management.

%% file: algorithm.tex
\section{The Algorithm} 
\label{sec:algorithm}

The algorithm is based on a shared array of size linear in $n$, where the array locations are associated to consecutive
indices.
A process \emph{registers} at a location by performing a successful \lit{test-and-set} operation on that location, and
\emph{releases} the location by re-setting the location to its initial value. A process performing a \lit{Collect}
simply reads the whole array in sequence. The challenge is to choose the locations at which the \Get{} operation
attempts to register so as to minimize contention and find a free location quickly. 

\begin{figure}[t]
\begin{center}
\includegraphics[scale=0.4]{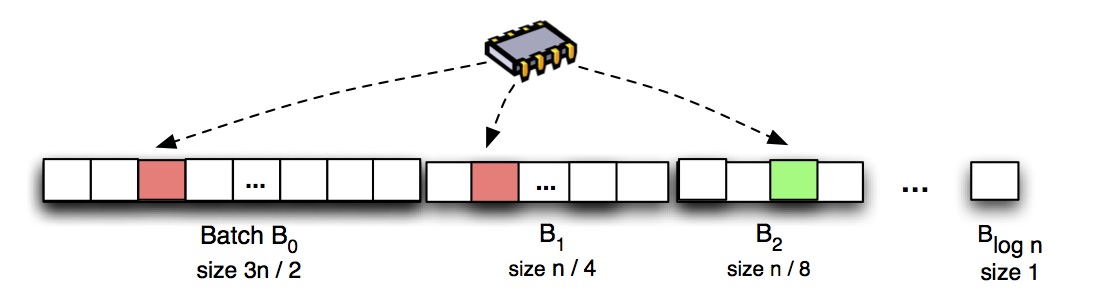}
\end{center}
\caption{An illustration of the algorithm's execution. A process probes locations in batches of increasing index, until
successful.}
\label{fig:algo}
\end{figure} 

Specifically, consider an array of size $2n$.\footnote{The algorithm 
works with small modifications for an array of size $(1 + \epsilon)n$, for $\epsilon > 0$ constant. This
yields a more complicated exposition without adding insight, therefore we exclusively consider the case $\epsilon =
1$.} 
The locations in the arrays are all initially set to $0$. 
The array is split into $\log n$ \emph{batches} $B_0, B_1, \ldots, B_{\log n - 1}$
such that $B_0$ consists of the first $\lfloor 3 n / 2 \rfloor$ memory locations, and each subsequent batch $B_i$ with
$i \geq 1$
consists of the first
$\lfloor n / 2^{i + 1} \rfloor$ entries after the end of batch $i - 1$. Clearly, this array is of size at 
most $2n$. (For simplicity, we omit the floor
notation in the following, assuming that $n$ is a power of two.) 

\lit{Get} is implemented as follows: the calling process accesses each batch $B_i$ in increasing order by index. In each batch $B_i$, the process sequentially attempts $c_i$ \lit{test-and-set} operations
on locations chosen uniformly at random from among all locations in $B_i$, where $c_i$ is a
constant. For the
analysis, it would suffice to consider $c_i = \kappa$ for all $i$, where the constant $\kappa$ is a uniform lower bound
of the $c_i$. We refrain from doing so in order to demonstrate which batches theoretically require higher values of
$c_i$. In particular, larger values of $c_i$ will be required to obtain high concentration bounds in later batches. In the implementation, we simply take $c_i = 1$ for all $i$. 

A process stops once it \emph{wins} a \lit{test-and-set} operation, and stores the index of the corresponding memory
location locally. When calling $\lit{Free}$, the process resets this location back to $0$. If, hypothetically, a
process reaches the last batch in the main array without stopping, losing all \lit{test-and-set} attempts, it will
proceed to probe sequentially all locations in a second backup array, of size exactly $n$. In this case, the process
would return $2n$ plus the index obtained from the second array as its value. Our analysis will show that the backup
is essentially never called. 

%% file: analysis-poly.tex

\section{Analysis}
\label{sec:analysis}

\paragraph{Preliminaries.} We fix an arbitrary
execution, and define the linearization order
for the \lit{Get} and \lit{Free} operations in the execution as follows. The linearization point for a \lit{Get}
operation is
given
by the time at which the successful \lit{test-and-set} operation occurs, while the linearization point for the
\lit{Free} procedure is given by the time at which the \lit{reset} operation occurs. (Notice that, since we assume a
hardware test-and-set implementation, the issues concerning linearizability in the context of randomization brought up
in~\cite{GHW11} are circumvented.)

\paragraph{Inputs and Executions.} We assume that the \emph{input} to each process is composed of four types of operations: \Get{}, \Free{}, \Collect{} and \Call{}, each as defined in Section~\ref{sec:model}. The \emph{schedule} is given by a string of process IDs, where the process ID appearing at the $i^{th}$ location in the schedule indicates which process takes a step at the $i^{th}$ time step of the execution. An \emph{execution} is characterized by the schedule together with the inputs to each process.

\paragraph{Running Time.} The correctness of the algorithm is straightforward. 
Therefore, for the rest of this section, we focus on the running time analysis.
We are interested in two parameters: the \emph{worst-case} running time i.e. the maximum number of probes that a
process
performs in order to register, and the \emph{average} running time, given by the expected number of probes performed
by an operation. We will look at these parameters first in polynomial-length executions, and then in
infinite-length executions.

Notice that the algorithm's execution is entirely specified by the schedule $\sigma$ (a series of process identifiers,
given by
the adversary), and by the processes' coin flips, unknown to the adversary when deciding the schedule. The
schedule is composed of low-level steps (shared-memory operations), which can be grouped into method calls. We say
that an event occurs at time $t$ in the execution if it occurs between steps $t$ and $t + 1$ in the
schedule $\sigma$. Let $\sigma_t$ be the $t$-th process identifier in the schedule. 
Further, the processes' random choices define a probability space, in which the algorithm's 
complexity is a random variable. 

Our analysis will focus on the first $\log \log n$ batches, as processes access later batches extremely rarely. Fix the constant 
$c = \max_k c_k$ to be the maximum number of trials in a batch. 
We say that a \lit{Get} operation
\emph{reaches} batch $B_j$ if it probes at least one location in the batch. 
For each batch index $j \in \{0, \ldots, \log \log n - 1\}$, we define $\pi_j$ to be $1$ for $j = 0$ and $1 / 2^{2^{j} +
5}$ for
$j \geq 1$, and $n_j$ to be $n$ if $j = 0$ and $n / 2^{2^{j} + 5}$ for $j \geq 1$, i.e., $n_j = \pi_j n$.
We now define properties of the probability distribution over batches, and of the array density.

\begin{definition}[Regular Operations]\label{def:regular}
We say that a $\lit{Get}$ operation is
\emph{regular} up to batch $0 \leq j \leq \log \log n - 1$, if, for any batch index $0 \leq k \leq j $, the probability
that the operation reaches batch $k$ is at most $\pi_k$. An operation is \emph{fully regular} if it is regular up to
batch $\log \log n - 1$. 
\end{definition}

\begin{definition}[Overcrowded Batches and Balanced Arrays]\label{def:balanced}
We say that a batch $j$ is \emph{overcrowded} at some time $t$ if at least $16 n_j = n / 2^{{2^j} + 1}$ distinct slots are occupied in batch 
$j$ at
time $t$. We say that the array is \emph{balanced} up to batch $j$ at time $t$ if none of the batches $0, \ldots, j$ are
overcrowded at time $t$. We say that the array is \emph{fully balanced} at time $t$ if it is balanced up to batch $\log
\log n - 1$. 
\end{definition}

In the following, we will
consider both polynomial-length executions and infinite executions. We will prove that in the first case, the array is
fully balanced throughout the execution with high probability, which will imply the complexity upper bounds. In the
second case, we show that the
data structure quickly returns to a fully balanced state even after becoming arbitrarily degraded, which implies low complexity for
most
operations. 

\subsection{Analysis of Polynomial-Length Executions}
\label{sec:poly}

We consider the complexity of the algorithm in executions consisting of $O( n^\alpha )$ 
\lit{Get} and \lit{Free} operations, where $\alpha \geq 1$ is a constant. A thread may have arbitrarily many \Call{} steps throughout its input. 
Our main claim is the following. 

\begin{theorem}
\label{thm:poly}
For $\alpha > 1$, given an arbitrary execution containing $O(n^\alpha)$ \lit{Get} and \lit{Free} operations, the
expected complexity of a \lit{Get} operation is constant, while its worst-case complexity is $O( \log \log n )$, with 
probability at least $1 - 1 / n^\gamma$, with $\gamma > 0$ constant. 
\end{theorem}

\noindent We now state a generalized version of the Chernoff bound that we will be using in the rest of this section. 
\begin{lemma}[Generalized Chernoff Bound~\cite{PS97}]
\label{lem:chernoff}
For $m \geq 1$, let $X_1, \ldots X_m$ be boolean random variables (not necessarily independent) with $\Pr[X_i = 1]
\leq p$, for all $i$. If, for any subset $S$ of $\{1, 2, \ldots, m\}$,  we have that $\Pr \left( \wedge_{i \in S}
X_i  \right) \leq p^{|S|}$, then we have that, 
for any $\delta > 0$, 
$$\Pr \left( \sum_{i = 1}^n X_i \geq (1 + \delta) np \right) \leq \left(\frac{e^{\delta}}{(1 + \delta)^{1 + \delta}}
\right)^{np}.$$
\end{lemma}

Returning to the proof, notice that the running time of a \lit{Get} operation is
influenced by the
probability of success of each of its \lit{test-and-set} operations. In  turn, this probability is influenced by the density of the current batch, which is related to the 
number of
\emph{previous} successful \lit{Get} operations that stopped in the batch. Our strategy is to show that the probability
that a \lit{Get} operation $\id{op}$ reaches a batch $B_j$ decreases doubly exponentially with the batch index
$j$. We prove this by induction on
the linearization time of the operation. Without loss of generality, assume that the execution contains exactly
$n^\alpha$ \Get{} operations, and let $t_1, t_2, \ldots, t_{n^{\alpha}}$ be the times in the execution when these \lit{Get}
operations are linearized. We first prove that \lit{Get} operations are fast while the array is in balanced
state.

\begin{prop}
\label{prop:prob}
 Consider a \lit{Get} operation $\id{op}$ and a batch index $0 \leq j \leq \log \log n - 2$. If at every time $t$ when
$\id{op}$ performs a random choice the Activity Array is balanced up to batch $j$, then, for any $1 \leq k \leq j + 1$,
the probability that $\id{op}$ reaches batch $k$ is at most $\pi_k$. This implies that the operation
is regular up to $j+1$. 
\end{prop}
\begin{proof}
	For $k = 1$, we 
upper bound the probability that the
operation does \emph{not} stop in batch $B_0$, i.e. fails all its trials in batch $B_0$. Consider the points in the
execution when the process performs
its trials in the batch $B_0$. At every such point, at most $n - 1$ locations in batch $B_0$ are occupied by other
processes, and at least $n / 2$ locations are always free. Therefore, the process always has probability at least $1 /
3$ of choosing an unoccupied slot in $B_0$ in each trial (recall that, since the adversary is oblivious, the scheduling
is independent of the random
choices). Conversely, the probability that the process fails all its $c_0$ trials in this batch is less than
$(2/3)^{c_0} \leq 1 / 2^7,$
for $c_0 \geq 16$, which implies the claim for $k = 1$. 

For batches $k \geq 2$, we consider the probability that the process fails all its trials in batch $k - 1$. Since, by
assumption, batch $B_{k - 1}$ is not overcrowded, there
are at
most $n / 2^{2^{k - 1} + 1}$ slots occupied in $B_{k - 1}$ while the process is performing random choices in this batch. On
the
other hand, $B_{k - 1}$ has $n / 2^{k}$ slots, by construction. Therefore, the probability that all of $p$'s trials fail
given that
the batch is not overcrowded is at most 
$$ \left( \frac{n / 2^{2^{k - 1} + 1}}{n / 2^{k}} \right)^{c_k} = \left(\frac{1}{2}\right)^{c_k (2^{k - 1} - k + 1)}.$$
The claim follows since $(1/2)^{c_k ( 2^{k - 1} - k + 1 )} \leq (1 / 2)^{2^{k} + 4} \leq \pi_k$ for $c_k \geq 16$ and $k
\geq 2$.
\end{proof}

We can use the fact that the adversary is oblivious to argue that the adversary cannot significantly 
increase the probability that a process holds a slot in a given batch. 
Due to space limitations, the full proof of the following lemma has been deferred to the Appendix. 
%

\input{time-invariant.tex}

\begin{proof}[Proof sketch] 
Fix a time $t$ and a process $q$. We first argue that, since the adversary is oblivious, the schedule must be fixed in advance and it suffices to consider only the times at which $q$ takes steps, i.e. times $t'$ where $\sigma_{t'} = q$. We denote $t_q[i]$ to be the $i^{th}$ step taken by $q$ in $\sigma$ and fix $\tau$ to be the index for which	 $t=t_q[\tau]$.

For \Get{} operation $G$ and time $t$, we define a random variable $S(G,t)$ to be the indicator variable for the event that the first step performed during $G$ occurs at time $t$. We also define the set $\mathcal{G}_x$ to be the set of all \Get{} operations which are followed by at least $x$ \Call{} steps in the input to $q$. Finally, we define a value $\alpha(\ell)$ such that $t_q[\alpha(\ell)]$ is the earliest time $t$ at which $S(G,t)$ may hold given that $G$ finishes executing at time $t_q[\ell]$ and returns a name in batch $j$. An exact expression for $\alpha(\ell)$ is given in the full proof.

Intuitively, in order for process $q$ to hold a name in batch $j$ at time $t$, there must exist a \Get{} operation, $G$ and a time $\ell\leq \tau$ such that $G$ completed at time $t_q[\ell]$ and $G$ is in the set $\mathcal{G}_{\tau-\ell}$, implying that the name returned by $G$ is not freed before time $t$. 
Using this idea, we prove that the probability of this occurring can be bounded above by the quantity:
\begin{eqnarray*} 
 \pi_j \sum_{m = 1}^{c_j} \sum_{\ell \leq \tau} \sum_{G\in\mathcal{G}_{\tau-\ell}} \Pr(S(G, t_q[\alpha(\ell)+m-1])).
 \end{eqnarray*}

Finally, we proved that the events $\{S(G,t_q[\alpha(\ell)+m-1])\}_{\ell,G}$ are mutually exclusive for fixed $m$, which reduces the above inequality to $$\Pr(B(q,t)=j)\leq \pi_j \sum_{m=1}^{c_j} 1 = \pi_jc_j,$$ completing the proof. 
\end{proof}

We can now use the fact that operations performed on balanced arrays are regular to show that balanced arrays are unlikely
to become unbalanced. In brief, we use Lemmas~\ref{lem:chernoff} and~\ref{timeinvariant} to obtain concentration bounds for the number of processes that may 
occupy slots in a batch $j$. This will show that any batch is unlikely to be overcrowded.

\begin{prop}
\label{chernoffapp}
Let $Q$ be the set of all processes. If, for all $q\in Q$ and some time $T$, the array was fully balanced at all times $t < T$, then for each $0
\leq j \leq \log \log n - 2$, batch $j$ is overcrowded at time $T$ with probability at most $(1/2)^{ \beta \sqrt{n} }$, where $\beta < 1$ is a constant.
\end{prop}
\begin{proof}
Let $\Pr_t^q(j)$ be the probability that process $q$ holds a slot in $B_j$ at some time $t$.
Applying Lemma~\ref{timeinvariant}, we have $\Pr_T^q(j) \leq c_j \pi_j$ for every $q\in Q$. 
For each $q \in Q$, let $X_j^q$ be the binary random variable with value $1$ if $q$ is in batch $j$ at time $T$,
and $0$ otherwise. The expectation of $X_j^q$ is at most $c_j \pi_j = c_j / 2^{2^{j} + 5}$.
Let the random
variable $X_j$ count the number of processes in batch $j$ at $T$. Clearly $X_j = \sum_{q} X_j^q$. 
By linearity of expectation, the expected value of $X_j$ is at most $c_j n_j = c_j n / 2^{2^{j} + 5}$. 
Next, we obtain a concentration bound for $X_{j}$ using Lemma~\ref{lem:chernoff}. 

It is important to note that the variables $X_{j}^q$ are not independent, and may be positively correlated in
general. For example, the fact that some process has reached a late batch (an improbable event), could imply that the
array is in a state that allows such an event, and therefore such an event may be more likely to happen again in that state. 

To circumvent this issue, we notice that, given the assumption that the array is fully balanced, and therefore balanced up to $j$, the probability
that any particular process holds a slot in $j$ must still be bounded above by $c_j \pi_{j}$, by Proposition~\ref{timeinvariant}. 
This holds given \emph{any} values of the random variables $\{X_j^q\}$ which are consistent with the array being balanced up to $j$. 
Formally, for any $R \subseteq Q$ with $q \notin R$, 
$\Pr \left( X_{j+1}^q | \wedge_{r \in R} X_{j+1}^r  \right)  \leq c_j \pi_{j+1}.$

\noindent In particular, for any $S=\{s_1,\dots,s_k\} \subseteq Q$ we have that 
\begin{eqnarray*}
\Pr \left( \wedge_{i \in S} X_{j+1}^i  \right) = \Pr\left( X_{j+1}^{s_1} \right) \cdot \Pr\left( X_{j+1}^{s_2} |
X_{j+1}^{s_1} \right) \cdot \ldots \\ \cdot \Pr\left(
X_{j+1}^{s_k} |
X_{j+1}^{s_1}, 
\ldots, X_{j+1}^{s_{k-1}}\right) 
\leq (c_j \pi_{j+1})^{|S|}.
\end{eqnarray*}

\noindent We can therefore apply Lemma~\ref{lem:chernoff} to obtain that
$\Pr \left( X_{j+1} \geq n / 2^{2^{j+1} + 1} \right)  
\leq \left( 1/2 \right)^{\beta \sqrt n}, $
 for $\beta < 1 / (45 \cdot 2^5)$, where we have used that $j \leq \log \log n - 1$ for the last inequality. This is the desired
bound.
\end{proof}

%

\noindent Next, we bound the probability that the array ever becomes unbalanced during the polynomial-length execution. 

\begin{prop}
\label{prop:overcrowded}
Let $t_i$ be the time step at which the $i^{th}$ \Get{} operation is linearized. For any $x\in \mathbb{N}$, the array is  fully balanced at every time $t$ in the interval $[0, t_x]$ with probability at least 
$1 - O( x \log\log n/ 2^{\beta \sqrt{n}} )$, where $\beta < 1$ is a constant.  
\end{prop}
\begin{proof}
	Notice that it is enough to consider times $t_i$ with $i = 1 \ldots x$, since \lit{Free} or
\lit{Collect} operations do not influence the claim.  We proceed by induction on the index $i$ of the \lit{Get}
operation in the linearization order. We prove that, for every $i = 1\ldots x$, the probability that there exists $t \leq t_i$ for which the
array is not fully balanced is at most $i / 2^{\beta \sqrt{n}}$.
	 
	For $i = 0$, the claim is straightforward, since the first operation gets a slot in batch $B_0$, so the array is
fully balanced at $t_1$ with probability $1$. 
	For $i \geq 1$, let $E_i$ be the event that, for some $\tau\leq t_i$, the array is not fully balanced at time $\tau$. From
the law of total probability we have that: $\Pr( E_i ) \leq \Pr( E_{i - 1} ) + \Pr( E_i | \neg E_{i - 1} ). $
	
	By the induction step, we have that $\Pr ( E_{i - 1} ) \leq (i - 1) \log\log n / 2^{\beta \sqrt n }$. We therefore need to
upper bound the term $\Pr( E_i | \neg E_{i - 1} )$, i.e. the probability that the data structure is not fully balanced
at time $t_{i}$ given that it was balanced at all times up to and including $t_{i - 1}$.
Proposition~\ref{chernoffapp} 
bounds the probability that a single batch is overcrowded at time $t_{i}$ by $(1/2)^{\beta \sqrt{n}}$. Applying the
union bound over the $\log\log n$ batches
gives $\Pr( E_i | \neg E_{i-1}) \leq \log\log n / 2^{\beta \sqrt{n} }.$ Thus, by induction 
$\Pr( E_i ) \leq \Pr( E_{i-1} ) + \Pr( E_i | \neg E_{i-1} ) \leq \frac{i\log\log n}{2^{\beta \sqrt{n}}},$
which proves
Proposition~\ref{prop:overcrowded}.
\end{proof}

\noindent Proposition~\ref{prop:overcrowded} has the following corollary for polynomial executions.

\begin{corollary}
 The array is balanced for the entirety of any execution of length $n^\alpha$ with probability at least 
$1 - O( n^\alpha \log\log n/ 2^{\beta \sqrt{n}} )$. 
\end{corollary}

\paragraph{The Stopping Argument.} The previous claim  shows that, during a polynomial-length execution,
we can practically assume that no batch is overcrowded. 
On the other hand, Proposition~\ref{prop:prob} gives an upper bound on the distribution over batches during such an
execution, given that no batch is overcrowded.

 To finish the proof of Theorem~\ref{thm:poly}, we combine the previous claims to lower bound the
probability that every
operation in an execution of length
$n^\alpha$ takes $O(\log \log n)$ steps by
$1 - 1 / n^\gamma$, with $\gamma \geq 1$ constant. Consider an arbitrary operation $\id{op}$ by process $p$ in such an
execution prefix. In order to take $\omega( \log \log n )$ steps, the operation must necessarily move past batch $\log
\log n - 1$ (since each process performs $c$ operations in each batch). We first bound the probability that
this event occurs assuming that no batch is overcrowded during the execution. 

Let $\ell = \log\log n -1$. By the assumption that no batch is overcrowded, we have in particular that there are at most
$16 n_\ell = n/2^{2^\ell +1} = \sqrt{n}/2$ processes currently holding names in $B_\ell$. Given that there are
less than $\sqrt n$ names occupied in batch $B_\ell$ at every time when $p$ makes a choice, the
probability that $p$ makes $c_\ell$ unsuccessful probes in $B_\ell$ is at most 
$ \left( \frac{\sqrt{n}}{n/2^{\ell+1}} \right)^{c_\ell} = \left( \frac{\log n}{\sqrt{n}} \right)^{c_\ell}.$

Therefore, by the union bound together with the law of total probability, the probability that \emph{any one} of the $n^{\alpha}$ operations in the execution
takes $\omega( \log \log n )$ steps is at most 
$n^{\alpha}\left( \left(\frac{\log n}{\sqrt n}\right)^{c_\ell} +
\frac{\log\log n}{2^{\beta \sqrt{n}}}  \right) \leq
\frac{1}{n^{\gamma}}, $
for $c_\ell \geq 2(\alpha + \gamma + 1)$, $\beta < 1$ constant, and large $n$. 
(Note that the number of probes $c_\ell$ is large enough to meet the requirements of Proposition~\ref{prop:prob}.) This concludes the proof of the high probability claim. The
expected step complexity claim follows from Proposition~\ref{prop:prob}.

\paragraph{Notes on the Argument.}
Notice that we can re-state the proof of Theorem~\ref{thm:poly} in terms of the step complexity of a single
\Get{} operation performing trials at times at which the array is fully balanced. 
\begin{corollary}
\label{cor:poly}
Consider a \Get{} operation with the property that, for any $0 \leq j \leq \log \log n - 1$,  the array is balanced up
to $j$ at all times when the operation performs trials in batch $j$. Then the step complexity of the operation is $O(
\log \log n )$ with probability at least $1 - O( 1/ n^{\gamma})$ with $\gamma \geq 1$ constant, and its expected step
complexity is constant. 
\end{corollary}

This raises an interesting question: 
at what point in the execution might \lit{Get} operations start taking $\omega(\log\log n)$ steps with probability
$\omega(1/n)$? Although we will provide a more satisfying answer to this question in the following section, the
arguments up to this point grant some preliminary insight.
Examining the proof, notice that a necessary condition for operations to exceed $O(\log\log n)$ worst case
complexity is that the array becomes unbalanced. By Proposition~\ref{prop:overcrowded} the probability that the array
becomes unbalanced at or before time $T$ is bounded by $O(T/2^{\sqrt{n}})$ (up to logarithmic terms). Therefore
operations cannot have $\omega(\log\log n)$ with non-negligible probability until $T=\Omega(
2^{\sqrt n})$.

%% file: time-invariant.tex

\begin{lemma}
\label{timeinvariant}
Suppose the array is fully balanced at all times $t < T$. Let $B(q,t)$ be a random variable whose value is equal to the batch in which process $q$ 
holds a slot at time $t$. Define $B(q,t)=-1$ if $q$ holds no slot at time $t$. Then for all $q,j,t < T$, $\Pr[B(q,t)=j] \leq c_j\pi_j$. 
\end{lemma}
\begin{proof}
We note that read operations cannot affect the value $\Pr[B(q,t)=j]$. For the purposes of this lemma, since the \Collect{} operation always performs exactly $n$ read operations and no others, we can assume without loss of generality that no \Collect{} operations appear in the input. This assumption is justified by replacing each \Collect{} operation with $n$ \Call{} operations before conducting the analysis.

We now introduce some notation. Let $t_q[i]$ be the index of $q$'s $i^{th}$ step in $\sigma$, and let $G_k$ 
be the $k^{th}$ \Get{} operation performed by $q$. Let $R(G)$ be the batch containing the name returned by operation $G$. Let $S(G,t)$ be the event that the first step executed during \Get{} operation $G$ occurs at time $t$. Similarly, let $C(G,t)$ be the event that the last step executed during $G$ occurs at time $t$. Intuitively $S(G,t)$ and $C(G,t)$ respectively correspond to the events that $G$ \emph{starts} at time $t$ or \emph{completes} at time $t$.

Because the schedule is fixed in advance, we have that, for any
value $x$ and any index $t$ satisfying $\sigma_{t} \neq q$, $\Pr(B(q,t)=x) = \Pr(B(q,t - 1)=x)$ holds. 
Therefore, it suffices to consider only the time steps at which $q$ acts. Hence we fix 
a time $t\leq T$ such that $\sigma_t=q$, and define $\tau$ such that $t=t_q[\tau]$.

 Note that, when executing some \Get{} $G_k$, each process performs exactly $c_b$ memory operations in batch $b$ before moving on to the next batch. We write $\hat{c}_j=\sum_{b < j}c_b$ which represents the maximum number of memory operations which can be performed in batches less than $j$.
 In particular, if $S(G_k,t_q[i])$ holds, then, assuming $G_k$ is still running at these times, the operations performed at times 
 $t_q[i+\hat{c}_j]$ until $t_q[i+\hat{c}_j+c_j-1]$ are performed on locations in $B_j$. 
 Indeed, if $G_k$ returns a name in $B_j$, then $C(G_k,t_q[i+\hat{c}_j+m])$ must hold for some $0\leq m < c_j$.

Now, consider the event $B(q,t)=j$. 
This event implies that there exists some \Get{} operation $G_k$ for which $R(G_k)=j$ and $C(G_k,t_q[i])$ holds where $i\leq \tau$
and $G_k$ is followed by at least $\tau-i$ \Call{} steps in the input of $q$. These conditions characterize the last \Get{} performed by $q$ before time $t$.

For convenience, we will let the set $\mathcal{G}_x$ be the set of \Get{} operations which are followed by at least $x$ 
\Call{} steps in the input. Note that, because the last \Get{} performed by $q$ before time $t$ is unique, the events $C(G_k,t_q[\ell])$ 
are mutually exclusive for all possible $k,\ell$ with $G_k\in \mathcal{G}_{\tau-\ell}$. With this in mind, we can write
\begin{eqnarray*} 
\Pr(B(q,t)=j) = \\ \sum_{\ell \leq \tau} \sum_{G\in\mathcal{G}_{\tau-\ell}} \Pr(C(G,t_q[\ell])\cap R(G)=j). 
\end{eqnarray*}

For convenience, we write $\alpha(\ell)=\ell-(\hat{c}_j+c_j)$. By our earlier observation, we know that if $G_k$ completed at time $\ell$ and $R(G)=j$, it must be that $G$ started somewhere in the interval $[t_q[\alpha(\ell)], t_q[\alpha(\ell)+c_j-1]]$. Thus, we can further rewrite 
\begin{eqnarray*} \Pr(B(q,t)=j) =  & \\ \sum_{m=1}^{c_j} \sum_{\ell \leq \tau} \sum_{G\in\mathcal{G}_{\tau-\ell}} \Pr(S(G,t_q[ \alpha(\ell) + m - 1])  \cap C(G,t_q[\ell])). &\end{eqnarray*}
By the law of total probability, this expression is then equivalent to 
\begin{eqnarray*}
\Pr(B(q,t)=j)  =  \\ \sum_{m = 1}^{c_j} \sum_{\ell \leq \tau} \sum_{G\in\mathcal{G}_{\tau-\ell}} \Pr(S(G,t_q[\alpha(\ell)+ m - 1])) \cdot \\ \cdot \Pr(C(G,t_q[\ell]) | S(G, t_q[\alpha(\ell)+m-1])).
\end{eqnarray*}
Given $S(G,t_q[\alpha(\ell)+m-1)$ for $1\leq m\leq c_j$, it must be that $C(G,t_q[\ell])$ implies $R(G)=j$, since the last trial performed by $G$ necessarily occurred in batch $j$. 
Applying this observation together with Proposition~\ref{prop:prob}, we have  $$\Pr(C(G,t_q[\ell]) | S(G, t_q[\alpha(\ell)+m-1]))\leq\pi_j.$$ 
We are left with 
\begin{equation}\label{eq:cpij} \Pr(B(q,t)=j)\leq \\ \pi_j \sum_{m = 1}^{c_j} \sum_{\ell \leq \tau} \sum_{G\in\mathcal{G}_{\tau-\ell}} \Pr(S(G, t_q[\alpha(\ell)+m-1])).\end{equation}

We claim that these events, $\{S(G,t_q[\alpha(\ell)+m-1])\}_{G, \ell}$ are mutually exclusive for fixed $m$. Specifically, we claim that (for fixed $m$) there cannot be an execution in which there are two pairs $(G, \ell) \neq (G',\ell')$ for which  $\ell,\ell' \leq \tau$, ${G\in\mathcal{G}_{\tau-\ell}}, {G'\in\mathcal{G}_{\tau-\ell'}}$, and $S(G, t_q[\alpha(\ell)+m-1])$, $S(G', t_q[\alpha(\ell')+m-1])$ hold. 

Suppose for the sake of contradiction that two such pairs do exist. Without loss of generality assume $\ell < \ell'$. Then $G'$ must appear after $G$ in the input to $q$. Furthermore, by assumption $G$ is followed by at least $\tau-\ell$ \Call{} operations and at least one \Free{} operation. Thus, the earliest time at which $G'$ could begin executing is if $G$ completes in exactly one step. In this case, $G$, the $\tau-\ell$ \Call{} operations, and the single \Free{} operation together take at least $\tau-\ell+2$ steps, 
and so $G'$ cannot possibly start before time 
$t_q[(\ell-(\hat{c}_j+c_j)+ m -1) + (\tau-\ell+2)] = t_q[\tau-(\hat{c}_j+c_j)+m+1]$. By assumption, $\tau \geq \ell'$, thus $t_q[\tau-(\hat{c}_j+c_j)+m+1] > t_q[\ell'-(\hat{c}_j+c_j)+m-1] = t_q[\alpha(\ell')+m-1],$ contradicting the assumed start time of $G'$.

Thus, $\{S(G,t_q[\alpha(\ell)+m-1])\}$ are mutually exclusive, as claimed.
Applying this observation, we upper bound the inner two summations of equation~(\ref{eq:cpij}) by $1$, and so equation~(\ref{eq:cpij}) reduces to
$\Pr(B(q,t)=j) \leq \pi_j\sum_{m=1}^{c_j} 1 = c_j\pi_j,$
which completes the proof.
\end{proof}

%% file: weak-infinity-analysis.tex

\subsection{Infinite Executions}
\label{sec:infty-analysis}

In the previous section, we have shown that the data structure ensures low step complexity in polynomial-length
executions. However, this argument does not prevent the data structure from reaching a bad state over
infinite-length executions. In fact, the adversary could in theory run the data structure until batches $B_1, B_2,
\ldots$ are overcrowded, and then ask a single process to \lit{Free} and \lit{Get} from this state
infinitely many times. The \emph{expected} step complexity of operations from this state is still constant, however the
expected \emph{worst-case} complexity would be logarithmic.
This line of reasoning motivates us to analyze the complexity of the data structure in infinite executions. Our analysis
makes the assumption that a thread releases a slot within polynomially many steps from the time when it acquired it.

\begin{definition} Given an infinite asynchronous schedule $\sigma$, we say that $\sigma$ is \emph{compact} if there
exists a constant $B \geq 0$ such that, for every time $t$ in $\sigma$ at which some process initiates a \lit{Get} method call, that
process executes a \lit{Free} method call at some time $t' < t+n^B$. 
\end{definition}

\noindent Our main claim is the following. 
\begin{theorem} \label{thm:infty} Given a compact schedule, every \Get{} operation on the \lit{LevelArray} will complete
in $O(\log\log n)$ steps with high probability.\end{theorem}

\paragraph{Proof Strategy.} We first prove that, from an arbitrary starting state, in particular from an unbalanced one,
the \lit{LevelArray} enters and remains in a \emph{fully balanced} state after at most polynomially many steps, with
high probability (see Lemma~\ref{lem:heal}). This implies that the array is fully balanced at \emph{any} given time with
high probability. 
The claim then follows by
Corollary~\ref{cor:poly}. 

\begin{lemma}
\label{lem:heal}
 Given a compact schedule with bound $B$ and a \lit{LevelArray} in arbitrary initial state, the \lit{LevelArray} will be fully balanced
after $n^{B}\log \log n$ total system steps with probability at least $1-O(n^B(\log\log n)^2/2^{\beta \sqrt{n}})$, where $\beta < 1$ is a constant. 
\end{lemma}
\begin{proof}
We proceed by induction on the batch index $j \geq 0$. Let $T_j$ be the interval $[j n^B, (j+1)n^B-1]$ of length
$n^B$ in the schedule. Let $Y_j$ be an indicator variable for the event that, for every $i \leq j$, the
\lit{LevelArray} is balanced up to $i$ throughout the interval $T_i$. Let $\beta < 1$ be the constant from Proposition~\ref{chernoffapp}. 
For convenience, we write $\mu = n^B\log\log n / 2^{\beta \sqrt{n}}$. 

We will show that at least one additional batch in the \lit{LevelArray} becomes balanced over each interval $T_j$. In
particular, we claim, by induction, that $\Pr(\neg Y_j) \leq j\mu$. 
In particular, the probability that the array fails to be balanced up to $j$ after interval $T_j$ is small. 
Our goal is to show that the \lit{LevelArray} is fully
balanced with probability $1-O( \mu\log\log n )$, which corresponds exactly to the inductive claim for $j=\log\log
n$.

For $j = 0$, the \lit{LevelArray} is always trivially balanced up to batch $0$ with probability 1. For the induction
step, we assume $\Pr(\neg Y_j) \leq j \mu$, and we show that the probability that the array is not balanced up to batch $j
+ 1$ over the interval $T_{j+1}$ is at most $(j+1)\mu$.
By the law of total probability,
$\Pr(\neg Y_{j+1}) \leq \Pr(\neg Y_j) + \Pr(\neg Y_{j+1} | Y_j). $

\noindent In particular, there are two reasons why $Y_{j+1}$ may fail to hold. Firstly, the \lit{LevelArray} may not be
balanced up to $j$ over the interval $T_j$, an event which is subsumed by $\neg Y_j$. Secondly, the \lit{LevelArray} may
become unbalanced in the interval $T_{j+1}$, despite having been balanced up to $j$ over $T_j$. By the inductive
hypothesis, $\Pr(\neg Y_j) \leq j\mu$. We will bound  $\Pr(\neg Y_{j+1} | Y_j)$ by applying the following claim to
intervals $T=T_j$ and $T'=T_{j+1}$:

\begin{claim}
Suppose the \LA{} is balanced up to batch $j \leq \log\log n - 2$ throughout some interval $T$ of length $n^B$. Let
$T'$ be the interval of length $n^B$ that follows $T$. Then the probability that \LA{} fails to be balanced up to
batch $j+1$ throughout $T'$ is at most $\mu$.
\end{claim}
\begin{proof}
Let $O$ be the set of \Get{} operations whose returned names are still held at the start of interval $T'$ (equivalently at the end of
$T$). Since the schedule is compact, each \Get{} in $O$ must have been initiated during the last $n^B$ steps, i.e., within $T$. 
This holds because the parent process of any \lit{Get} initiated before $T$ would have been required to call 
\lit{Free} \emph{before} the end of $T$. Therefore, all decision points of every \Get{} in $O$ must have occurred in $T$. 
Thus, by the initial assumption, the precondition of Proposition~\ref{chernoffapp} is satisfied, and we have, for each $i \leq j+1, t \in
T'$, the probability that batch $i$ is overcrowded is at most $(1/2)^{\beta \sqrt{n}}$. 

By the union bound, the probability that any batch $i \leq j+1$ becomes overcrowded at any time $t \in T'$ is thus at
most $(j+1)|T'|/2^{\beta \sqrt{n}} \leq \mu$, proving the claim. 
\end{proof}

Returning to the proof of Lemma~\ref{lem:heal}, we have
$ \Pr(\neg Y_j) + \Pr(\neg Y_{j+1} | Y_j) \leq j\mu + \mu = (j+1)\mu,$
\noindent as desired.
\end{proof}
\noindent Lemma~\ref{lem:heal} naturally leads to the following corollary, whose proof can be found in the Appendix. 


\begin{corollary}
\label{corbadstate}
 For any time $t \geq 0$ in the schedule, the probability that the array is not fully balanced at time $t$ is at most
$\mu$.
\end{corollary}
\begin{proof}
 Pick an arbitrary time $t \geq 0$ in the schedule. If $t \leq n^{B} \log \log n$, then the claim follows from
Proposition~\ref{prop:overcrowded}. Otherwise, fix $S_{\id{init}}$ to be the initial state at time $t - n^{B} \log
\log n$, and apply Lemma~\ref{lem:heal} at this state, to obtain that the array is fully balanced with probability at
least $1 - \mu$ at $t$, as desired. 
\end{proof}

To complete the proof of Theorem~\ref{thm:infty}, fix an arbitrary compact schedule $\sigma$, and an arbitrary \Get{}
operation $\id{op}$ in the schedule. We upper bound the probability that $\id{op}$ takes $\omega( \log \log n )$ steps.
We know that, in the worst case, the operation performs $O( n )$ total steps (including steps in the backup). Let
$t_0, t_1, \ldots, t_k$ be the times in the execution when the operation performs random probes. By the structure of
the algorithm, $k = O( \log n )$. Our goal is to prove that $k = O ( \log \log n )$, with high probability. 

First, by Corollary~\ref{corbadstate} and the union bound, the probability that the array is not fully balanced at any
one of the times $\{t_i\}_{i = 1\ldots k}$ is at most $k \mu$. 
Assuming that the array is fully balanced at all times $t_i$, the probability that the process takes $\omega( \log \log
n )$ steps is at most $1 / n^{c}$, for $c \geq 1$, by Corollary~\ref{cor:poly}. Recall that $\mu$ is exponentially small in $n$. Then by the law of total
probability, the probability that the operation takes $\omega( \log \log n )$ steps is at most $k \mu + 1 /
n^{c} = O( 1 / n^\gamma )$, for $\gamma = c \geq 1$. Inversely, an arbitrary operation takes $O( \log \log n )$ steps in a compact schedule with high
probability, as claimed. The expectation bound follows similarly.

%% file: results.tex
\vspace{-1em}
\section{Implementation Results}
\label{sec:results}

\begin{figure*}[t]
\begin{center}
\begin{subfigure}[b]{0.45\textwidth}
 \centering 
 \includegraphics[width=2.7in, height=1.5in]{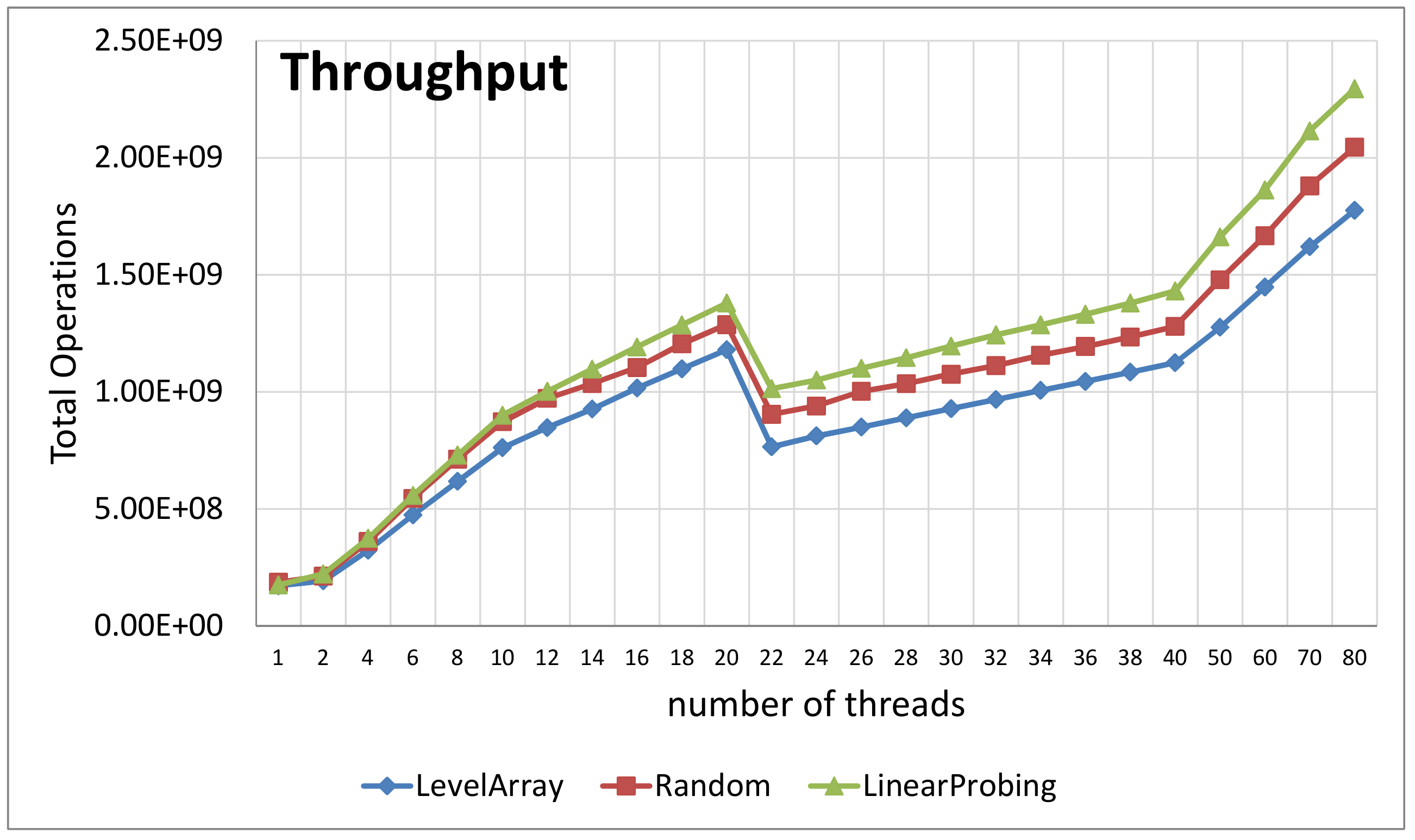}
 \label{fig:throughput}
\end{subfigure}
\quad
\begin{subfigure}[b]{0.45\textwidth}
 \centering 
 \includegraphics[width=2.7in, height=1.5in]{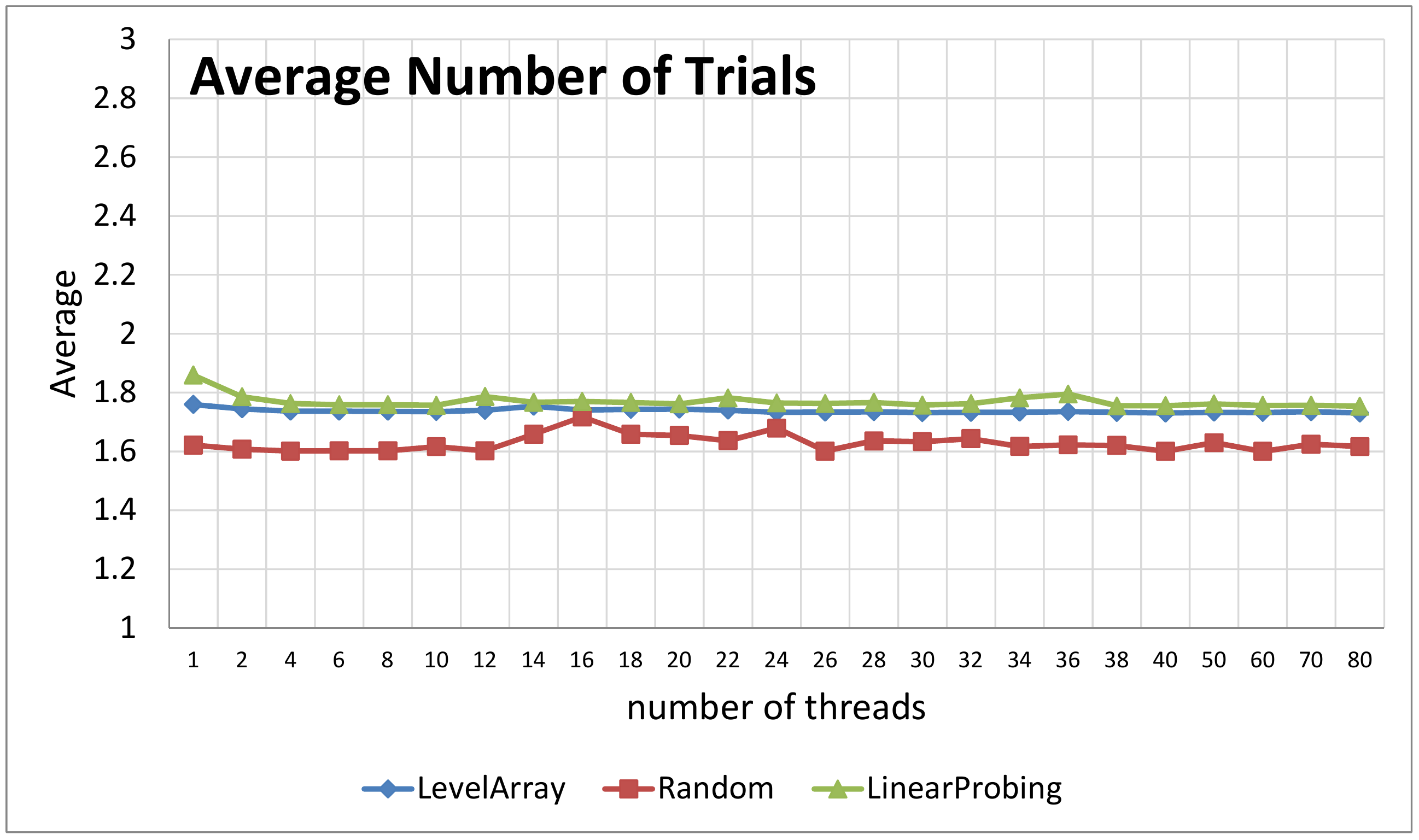}
 \label{fig:average}
\end{subfigure}

\begin{subfigure}[b]{0.45\textwidth}
 \centering 
 \includegraphics[width=2.7in, height=1.5in]{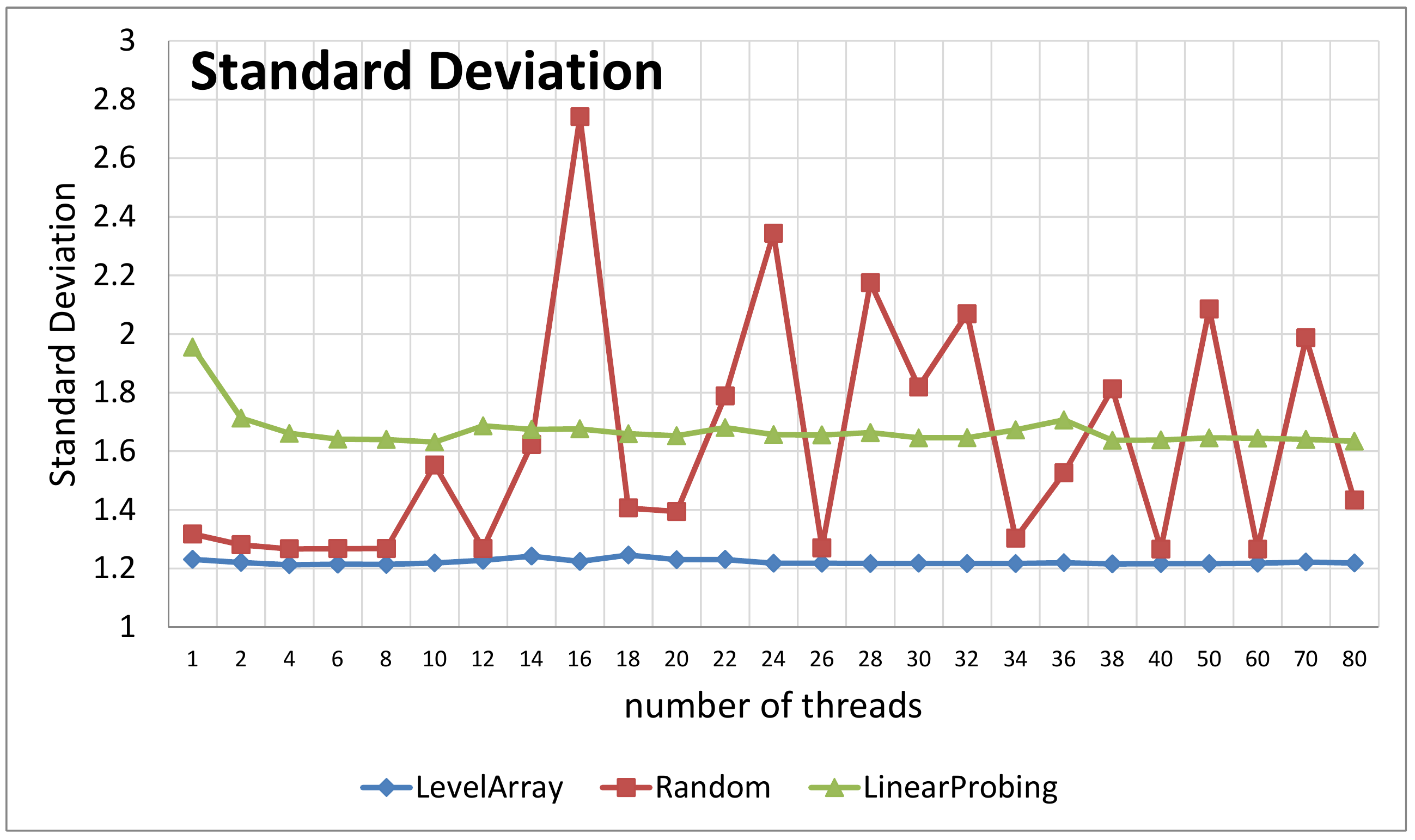}
 \label{fig:STD}
\end{subfigure}
\quad
\begin{subfigure}[b]{0.45\textwidth}
 \centering 
 \includegraphics[width=2.7in, height=1.5in]{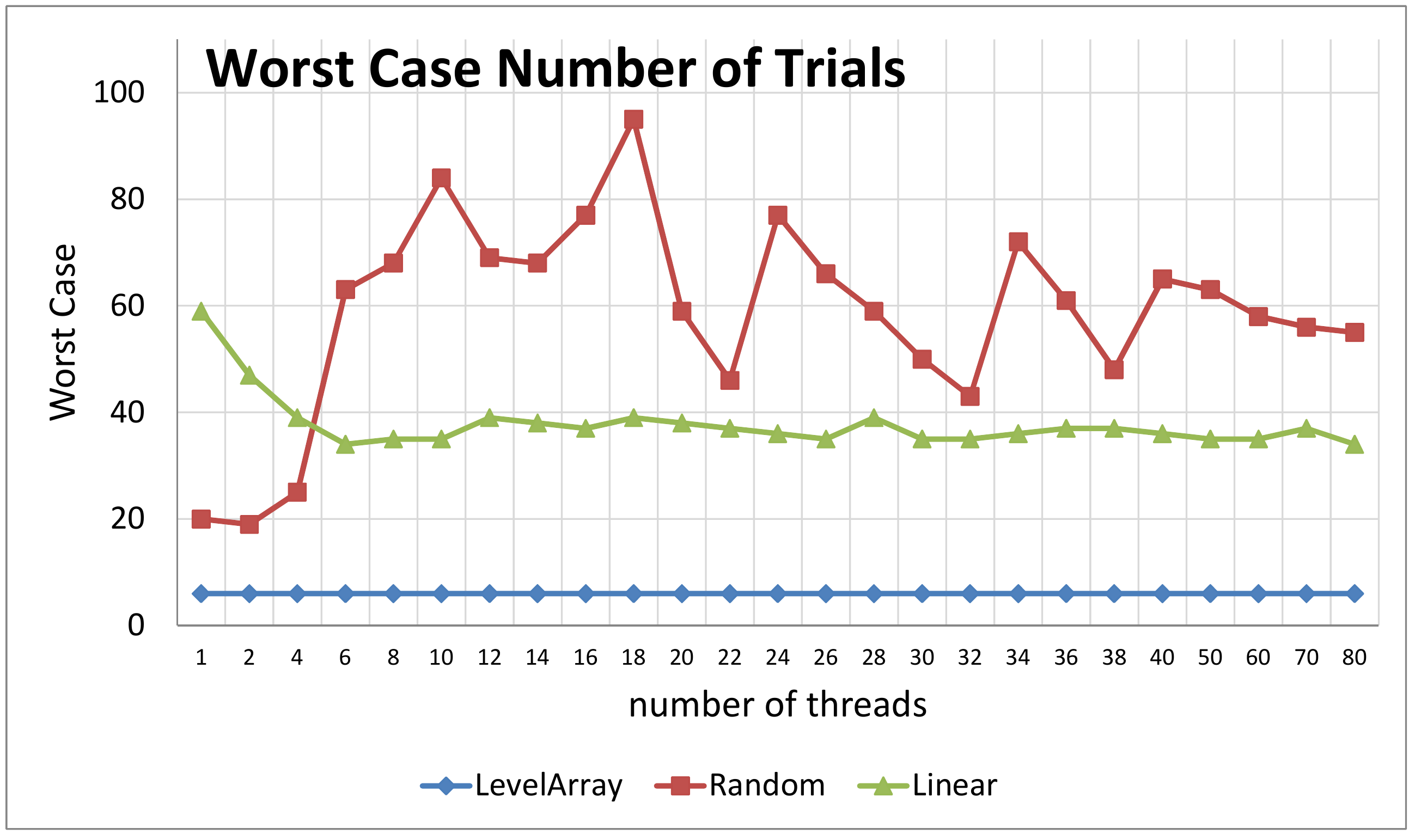}
 \label{fig:worst}
\end{subfigure}
\end{center}
\caption{Comparing the performance of \lit{LevelArray} with \lit{Random} and \lit{LinearProbing}. Throughput and average
complexity are similar, while the \lit{LevelArray} is significantly more stable in terms of standard deviation and
worst-case complexity.}
\label{fig:hash}
\end{figure*}

%
%
%

 \paragraph{Methodology.} The machine we use for testing is a Fujitsu PRIMERGY RX600 S6 server with four Intel Xeon
E7-4870 (Westmere EX) processors. Each processor has 10 2.40 GHz cores, each of which multiplexes two hardware threads,
so in total our system supports $80$ hardware threads. Each core has private write-back L1 and L2 caches; an inclusive
L3 cache is shared by all cores.

We examine specific behaviors by adjusting the following benchmark parameters. The parameter $n$ is the number of hardware threads spawned, while
$N$ is the maximum number of array locations that may be registered at the same time. For $N > n$, we
emulate concurrency by requiring each thread to register $N / n$ times before deregistering. The parameter $L$ is the
number of slots in the array. In our tests, we consider values of $L$ between $2N$ and $4N$. 

The benchmark first allocates a global array of length $L$, split as described in Section~\ref{sec:algorithm}, and then
spawns $n$ threads which repeatedly register and deregister from the array. In the implementation, threads perform
exactly \emph{one} trial in each batch, i.e. $c_\ell = 1$, for all batches $\ell = 1, \ldots, \log N$. (We tested the
algorithm with values $c_\ell > 1$ and found the general behavior to be similar; its performance is slightly lower given
the extra calls in each batch. The relatively high values of $c_\ell$ in the analysis are justified since our objective was to obtain high concentration bounds.) 
Threads use \lit{compare-and-swap} to acquire a location. We used the Marsaglia and
Park-Miller (Lehmer) random number
generators, alternatively, and found no  difference between the results.

The \emph{pre-fill percentage} defines the percentage of array slots that are occupied during the execution we examine.
For example, $90$\% pre-fill percentage causes every thread to perform $90$\% of its registers \emph{before} executing
the main loop, without deregistering. Then, every thread's main-loop performs the remaining $10$\% of the register and
deregister operations, which execute on an array that is  $90$\% loaded at every point. 

In general, we considered regular-use parameter values; we considered somewhat exaggerated contention levels (e.g. $90$\% pre-fill percentage) 
since we are interested in the worst-case behavior of the algorithm. 

\paragraph{Algorithms.} We compared the performance of \lit{LevelArray} to three other common algorithms used for fast registration. 
The first alternative, called \lit{Random}, performs trials at random in an array of the same size as our algorithm, until
successful. The second, called \lit{LinearProbing}, picks a random location in an array and probes locations linearly to
the right from that location, until successful. We also tested the \emph{deterministic} implementation that starts at
the first index in the array and probes linearly to the right. Its average performance is at least two orders of
magnitude worse than all other implementations for all measures considered, therefore it is not shown on the graphs.

\paragraph{Performance.} 
Our first set of tests is designed to determine the throughput of the algorithm, i.e. the total number of \lit{Get} and
\lit{Free} operations that can be performed during a fixed time interval. We analyzed the throughput for values of $n$
between $1$ and $80$, requiring the threads to register a total number of $N$ emulated threads on an array of size $L$ 
between $2N$ and $4N$. We also considered the way in which the throughput is affected by the different pre-fill
percentages. 

Figure~\ref{fig:hash} presents the results for $n$ between $1$ and $80$, $N = 1000 n$ simulated operations, $L = 2N$,
and a pre-fill percentage of $50$\%. We ran the experiment for $10$ seconds, during which time the algorithm performed
between $200$ million and $2$ billion operations. The first graph gives the total number of successful operations as a
function of the number of threads. 
As expected, this number grows linearly with the number of threads. (The variation at $20$ is because this is the
point where a new processor is used---we start to pay for the expensive inter-processor communication. Also, notice that the X
axis is not linear.) The fact that the throughput of \lit{LevelArray} is lower than that of \lit{Random} and
\lit{LinearProbing} is to be expected, since the average number of trials for a thread probing randomly is lower than
for our algorithm (since \lit{Random} and \lit{LinearProbing} use more space for the first trial, they are more likely
to succeed in one operation; \lit{LinearProbing} also takes advantage of better cache performance.)
This fact is illustrated in the second graph, which plots the \emph{average} number of trials per operation. For all 
algorithms, the average number of trials per \lit{Get} operation is between $1.5$ and $1.9$. 

The lower two graphs illustrate the main weakness of the simple randomized approaches, and the key property
of \lit{LevelArray}. In \lit{Random} and \lit{LinearProbing}, even though processes perform very few trials on average,
there are always some processes that have to perform a large number of probes before getting a location. Consequently,
the standard deviation is high, as is the worst-case number of steps that an operation may have to take. (To decrease the impact of outlier executions, 
the worst-case shown is averaged over all processes, and over several repetitions.) On the other
hand, the \lit{LevelArray} algorithm has predictable low cost even in extremely long executions. In this case, the
maximum number of steps an operation must take before registering is at most $6$, taken over $200$ million to $2$
billion 
operations. The results are similar for pre-fill percentages between $0$\% and $90$\%, and for different array sizes.
These bounds are also maintained in executions with more than $10$ billion operations.

\begin{figure}[t]
\begin{center}
\includegraphics[width=3.25in, height=1.85in]{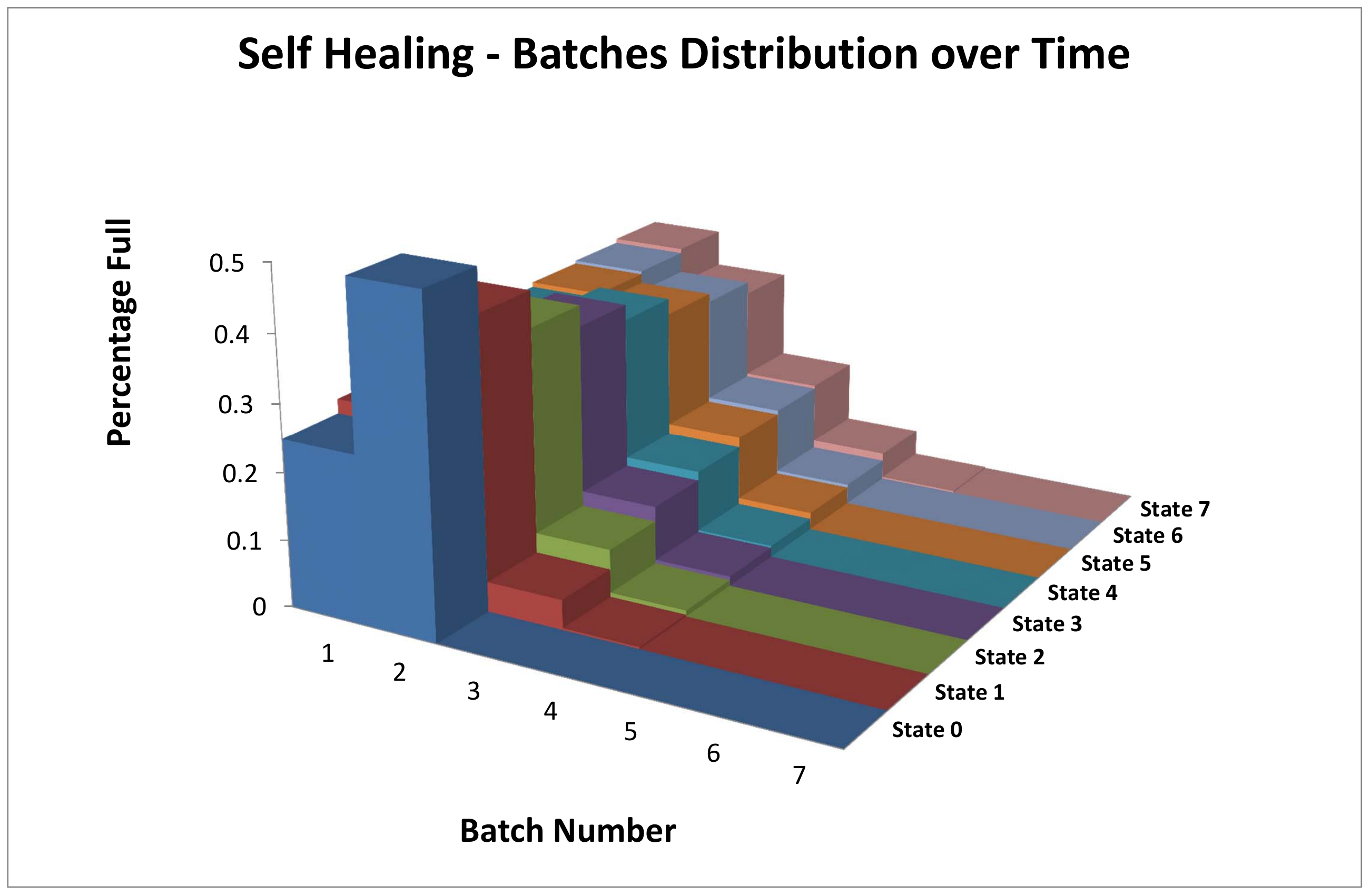}
\end{center}
\caption{The healing property of the algorithm. The array starts in unbalanced state (batch two is overcrowded),
and smoothly transitions towards a balanced state as more operations execute. Snapshots are taken every $4000$
operations.}
\label{fig:histogram}
\end{figure} 
\paragraph{The Healing Property.} The stable worst-case performance of \lit{LevelArray} is given by the properties
of the distribution of probes over batches. However, over long executions, this distribution might get 
skewed, affecting the performance of the data structure. The analysis in Section~\ref{sec:infty-analysis}
suggests that the batch distribution returns to normal after polynomially many operations. We test this
argument in the next experiment, whose results are given in Figure~\ref{fig:histogram}. 

The figure depicts the distribution of threads in batches at different points in the execution. Initially, the first
batch is a quarter full, while the second batch is half full, therefore overcrowded. As
we schedule operations, we see that the distribution returns to normal. After approximately $32000$ arbitrarily chosen
operations are scheduled, the distribution is in a stable state. (Snapshots are taken every $4000$ operations.)
The speed of convergence is higher than predicted by the analysis. We obtained the same results for variations of the
parameters. 

%% file: conclusion.tex
\vspace{-1em}
\section{Conclusions and Future Work}
\label{sec:conclusion}

In general, an obstacle to the adoption of randomized algorithms in a concurrent
setting is the fact that, while their performance may be good on
average, it can have high variance for individual threads in long-lived executions. Thus,
randomized algorithms are seen as unpredictable. In this paper, we
exhibit a randomized algorithm which combines the best of
both worlds, guaranteeing good performance on average and in the worst
case, over long finite or even infinite executions (under reasonable schedule assumptions). 
One direction of future work would be investigating randomized solutions with the same
strong guarantees for other practical concurrent problems, such as elimination~\cite{ShavitT97} or
rendez-vous~\cite{AHM11}.

%% file: biblio.tex
\bibliographystyle{plain}
\bibliography{bibliography}